\documentclass[runningheads]{llncs}
\usepackage[T1]{fontenc}
\usepackage{graphicx}
 
\usepackage{amsmath} 
\usepackage{amssymb} 
\usepackage{stmaryrd} 

\usepackage{varwidth}
\usepackage{listings} 
	\newcommand{\snip}[1]{\mbox{\texttt{#1}}} 

\usepackage{multirow}

\newcommand{\var}[1]{\mathit{#1}} 
\newcommand{\func}[1]{\mathit{#1}} 
\newcommand{\sdef}{\:\:\widehat{=}\:\:} 
\newcommand{\powset}{\mathcal{P}} 

\newcommand{\T}{\mathcal{T}} 
\newcommand{\V}{\mathcal{V}} 
\newcommand{\s}{\mathcal{S}} 

\newcommand{\D}{\mathcal{D}} 
\newcommand{\I}{\mathcal{I}} 

\newcommand{\VP}{\V}

\newcommand{\instr}{\var{inst}}
\newcommand{\cond}{\var{cond}}

\newcommand{\dsqsubseteq}{\sqsubseteq^\sharp}
\newcommand{\dsqsupseteq}{\sqsupseteq^\sharp}
\newcommand{\dgamma}{\gamma^\sharp}
\newcommand{\dbot}{\bot^\sharp}
\newcommand{\dtop}{\top^\sharp}
\newcommand{\djoin}{\sqcup^\sharp}
\newcommand{\dmeet}{\sqcap^\sharp}
\newcommand{\bigdjoin}{\sideset{}{^\sharp}\bigsqcup}
\newcommand{\bigdmeet}{\sideset{}{^\sharp}\bigsqcap}
\newcommand{\bd}{\mathbb{B}^\sharp}
\newcommand{\sd}{\mathbb{S}^\sharp}

\newcommand{\sconc}{\func{exec}}

\newcommand{\isqsupseteq}{\sqsupseteq^\times}
\newcommand{\isqsubseteq}{\sqsubseteq^\times}
\newcommand{\igamma}{\gamma^\times}
\newcommand{\ibot}{\bot^\times}
\newcommand{\itop}{\top^\times}
\newcommand{\ijoin}{\sqcup^\times}
\newcommand{\imeet}{\sqcap^\times}
\newcommand{\bigijoin}{\sideset{}{^\times}\bigsqcup}

\newcommand{\transitions}{\func{transitions}}
\newcommand{\stabilise}{\func{stabilise}}
\newcommand{\close}{\func{close}}

\newcommand{\rely}{\func{rely}}

\newcommand{\assign}{\var{Assign}}

\newcommand{\havocSymb}{\lightning}
\newcommand{\havoc}[2]{#1\havocSymb#2}

\newcommand{\lstingSpace}{\hspace{0.08\textwidth}}

\begin{document}

\title{Generating Rely-Guarantee Conditions with the Conditional-Writes Domain}

\author{James Tobler\orcidID{0000-0002-1205-3455} \and
Graeme Smith\orcidID{0000-0003-1019-4761}}

\institute{Defence Science and Technology Group, Australia \\
School of Electrical Engineering and Computer Science,\\
The University of Queensland, Brisbane, Australia\\
\email{james.tobler@uq.edu.au}}

\maketitle

\begin{abstract}
Abstract interpretation has been shown to be a promising technique for the thread-modular verification of concurrent programs. Central to this is the generation of interferences, in the form of rely-guarantee conditions, conforming to a user-chosen structure. In this work, we introduce one such structure called the \emph{conditional-writes} domain, designed for programs where it suffices to establish only the \emph{conditions} under which particular variables are written to by each thread. We formalise our analysis within a novel abstract interpretation framework that is highly modular and can be easily extended to capture other structures for rely-guarantee conditions. We formalise two versions of our approach and evaluate their implementations on a simple programming language.

\end{abstract}

\section{Introduction}

Abstract interpretation has been shown to be a scalable technique for automatically verifying sequential programs~\cite{cousot77,mine17}. Its distinctive advantage lies in the provision of a simple parameter called an \emph{abstract domain}, which constrains the structure of generated intermediary assertions over the program state. This allows proofs to be generated at a large scale, given a suitable abstract domain is provided for the target program. Unfortunately, these are not always easy for users to infer.

For concurrent programs, \emph{interferences} between threads, captured by rely-guarantee conditions, must also be generated. Our work is based on the conjecture that, contrary to intermediary assertions (required for sequential programs), many of the interferences that suffice for verification (of concurrent programs) tend to take a similar structure across various concurrent programs. In this paper, we capture one of these structures in an abstract domain called the \emph{conditional-writes} domain, which is designed for cases where it suffices to reason about thread interference in terms of only the \emph{conditions} under which variables may be updated, rather than what they are updated to. For example, where it suffices to establish: ``Thread $A$ can write to $x$ when $z=0$ and can write to $y$ when $z=1$ but never writes to $z$.''

To precisely capture this structure in an abstract domain, we introduce a novel framework for applying abstract interpretation to concurrent programs. In current frameworks such as~\cite{monat17}, thread interferences are abstracted into numeric abstract domains over the program variables and their primed counterparts, where a primed variable represents its value in the post-state of the interference. This approach is limited in that primed and unprimed variables are treated equally --- that is, whether a variable is primed has no impact on where it may be placed within the structure of an abstract domain element. Consequently, they cannot represent structures that place different constraints on pre-states and post-states, such as our conditional-writes domain. This has motivated us to introduce a new class of abstract domain called an \emph{interference domain}, of which the conditional-writes domain is one example.

In the abstract interpretation framework, a traditional abstract domain combines a lattice $\D$ over sets of states together with a \emph{transfer function} that derives an element of $\D$ as a conservative pre- or post-condition of a program instruction. Conversely, an interference domain combines a lattice $\I$ over sets of \emph{pairs of states} (representing rely and guarantee conditions) with two functions, both parameterised by an abstract domain with lattice $\D$. The first, which we call $\transitions$, maps some $d\in\D$ and an assignment to an element of $\I$ that over-approximates the set of state transitions that may be induced by executing the given assignment under any state captured by $d$. The second, which we call $\stabilise$, maps some $d\in\D$ and $i\in\I$ to an element of $\D$ that over-approximates the set of states reachable within one step of $i$ from any initial state in $d$. Together, these two functions can be used to respectively derive and apply thread interferences (guarantee and rely conditions, respectively) during a thread-modular analysis.

In addition to these two functions, we equip our conditional-writes domain with a function $\close$ that over-approximates the transitive closure of thread interferences. Using this function to make some $i\in\I$ transitive allows us to account for the effect of one or more transitions in $i$ with just one call to $\stabilise$, thereby simplifying the process of applying interference at the cost of deriving this transitive closure. We determine whether this trade-off is worthwhile by implementing and evaluating two versions of our analysis: one where transitivity is enforced, and one where it is not.

We begin in Section~\ref{sec:lang} by introducing a simple programming language on which to evaluate our analyses. In Section~\ref{sec:abs}, we define an abstract state domain which is used as a parameter to our analyses, and in Section~\ref{sec:cond}, we define our interference domain including the functions $\transitions$, $\stabilise$, and $\close$. The analyses with and without transitivity are formalised in Section~\ref{sec:analysis} and evaluated on the simple programming language in Section~\ref{sec:eval}. We discuss related work in Section~\ref{sec:rel} before concluding in Section~\ref{sec:con}. 

\section{Programming Language}
\label{sec:lang}

We define $\VP$ to be the set of program variables, $\T$ the identifiers of the threads comprising the program, and $\s$ the set of all states, where a state maps each variable in $\VP$ to a constant. Each thread $t\in\T$ represents a sequential program in the language $\instr$:
\begin{align*}
	\instr::=\ 
	&\instr\:\snip{;}\:\instr\:\mid\: x_1,...,x_n\snip{ := }e_1,...,e_n\:\mid\:\snip{if }\cond\snip{ then }\instr\snip{ else }\instr\\
	&\mid\:\snip{while }\cond\snip{ do }\instr\:\mid\: \snip{skip}
\end{align*}
where $\cond$ is the set of all boolean expressions, and for each $i\in\{1,...,n\}$, we have $x_i\in\VP$ and $e_i$ is an expression. The concrete semantics of $\instr$ is denoted:
\[
	\sconc:\instr\times\s\to\s
\]
where $\sconc(c,s)$ is the program state reached by executing instruction $c$ in the state $s$. In our simple language, all assignment statements are atomic, as are the evaluations of all branch conditions and loop guards.

For each thread $t\in\T$, we also let $\V_t\subseteq\V$ represent the set of variables over which the rely condition of $t$ should be defined. Using $\V_t$, we can avoid complicating the intermediary assertions that are generated for $t$ when accounting for interference specified by the rely. Typically, $\V_t$ is set to the union of the global variables and the local variables of $t$. However, in some cases, precision can be improved by including the local variables of other threads in $\V_t$\footnote{In standard rely-guarantee approaches, this is akin to converting a local variable into an ``auxiliary variable'' \cite{sto91,xu97}.}, and conversely, efficiency can be improved by restricting $\V_t$ to only those variables that are relevant to our particular verification goals. Most existing abstract interpretation frameworks for concurrent programs do not distinguish between local and global variables~\cite{monat17,mine14,mine12} and therefore effectively have $\V_t=\V$ for all $t\in\T$.

\section{Abstract State Domain}
\label{sec:abs}

Our analysis is parameterised by an abstract domain with lattice $\D$, representing sets of states. To distinguish this from our interference domain, we call it the \emph{state domain}, and call $\D$ the \emph{state lattice}. For $d_1,d_2\in\D$, we denote by $d_1\dsqsubseteq d_2$ that $d_1$ is less than or equal to $d_2$ on the lattice, meaning it represents the same or fewer states. (From now on, we may say that $d_1$ is ``stronger'' than $d_2$, or that $d_2$ is ``weaker'' than $d_1$.) The lattice's top element $\dtop$ represents the set of all states, and the bottom element $\dbot$ represents the empty set.
In addition to the above, and the usual join $\djoin$ (least upper bound) and meet $\dmeet$ (greatest lower bound) operators, our state domain is equipped with the following basic functions:
\begin{description}
	\item $\dgamma:\D\to\powset(\s)$ where $\dgamma(d)$ is the set of states represented by $d$.
	\item $\sd:\instr\times\D\to\D$ where $\sd(c,d)$ is the abstract post-state of instruction $c$ with respect to the abstract pre-state $d$.
	\item $\bd:\cond\times\D\to\D$ where $\bd(b,d)$ is the abstract state capturing all states in $d$ that satisfy the condition $b$.
	\item $\havocSymb:\D\times\powset(\VP)\to\D$ where $\havoc{d}{V}$ is an abstract state weaker than $d$ which at least includes the set of states obtained from $d$ when the value of all variables in $V$ are unconstrained. Formally:	
	\[
		\dgamma(\havoc{d}{V})\supseteq\left\{
			s\in\s\mid
			\exists s_0\in\dgamma(d):
				\forall v\in\VP\setminus V:s(v)=s_0(v)
		\right\}.
	\]
	Similar to the above functions, the precise definition of $\havocSymb$ is provided by the specific abstract state domain given to our analysis. However, soundness of our analysis requires that this definition satisfies the above inequality, in addition to the basic properties:
	\begin{enumerate}
		\item Havocing any variable set in $\dbot$ results in $\dbot$:
		\[
			\forall V\in\powset(\VP):\havoc{\dbot}{V}=\dbot.
		\]
		\item Havocing additional variables that are not constrained has no effect:
		\[
			\forall V_1,V_2\in\powset(\V),d\in\D:\havoc{d}{V_1}=d\implies\havoc{d}{V_1\cup V_2}=\havoc{d}{V_2}.
		\]
		\item The $\havocSymb$ operator is monotonic in that:
		\[
			\forall V_1,V_2\in\powset(\V),d_1,d_2\in\D:V_1\subseteq V_2\land d_1\dsqsubseteq d_2\implies\havoc{d_1}{V_1}\dsqsubseteq\havoc{d_2}{V_2}.
		\]
	\end{enumerate}
	We give $\havocSymb$ the highest precedence of all the state lattice operators.
\end{description}

\section{Conditional-Writes Domain}
\label{sec:cond}

To capture interference between threads, we introduce the notion of an \emph{interference domain}. An interference domain is defined by a lattice $\I$ whose elements represent sets of pairs of states. It defines two functions over this lattice, both parameterised by a compatible state domain with lattice $\D$:
\begin{align*}
	\stabilise:\I\times\D\to\D
\quad&&\quad
	\transitions:\D\times\assign\to\I
\end{align*}
where $\assign$ are those elements of $\var{inst}$ corresponding to assignment statements. We refer to a state domain as \emph{compatible} if it defines on $\D$ all functions used in the definitions of $\stabilise$ and $\transitions$.

In this work, we define an interference domain called the \emph{conditional-writes domain}. It is parameterised by any state domain described by Section~\ref{sec:abs} and defines $\I$ to be a lattice over mappings from program variables to elements of the state lattice $\D$. Formally:
\[
	\I=\V\to\D.
\]
Intuitively, an element $i\in\I$ of this domain maps each variable to an over-approximation of the set of states under which that variable may be written to, which we call its \emph{write-condition}. For example, if $i$ represents a valid guarantee condition of thread $t$, and takes the form $i=[v_1\mapsto P_1,\:v_2\mapsto P_2]$, then $t$ only updates $v_1$ when the program is in a state captured by $P_1$, and only updates $v_2$ when the program is in a state captured by $P_2$. If the program is in a state captured by both $P_1$ and $P_2$, then $t$ may change either variable, or both variables in one step. This is captured by the concretisation function $\igamma:\I\to\powset(\s\times\s)$:
\[
	\igamma(i)\sdef\left\{
	\left(s_1,s_2\right)\in\s\times\s\mid
	\forall v\in\VP:s_2(v)\ne s_1(v)\implies s_1\in\dgamma(i(v))
	\right\}.
\]
That is, the pre-state of any transition in $\igamma(i)$ which changes variable $v$ must be a state in which $i$ allows $v$ to be updated.

Similar to our abstract state domain, for $i_1,i_2\in\I$, we denote by $i_1\isqsubseteq i_2$ that $i_1$ is less than or equal to $i_2$ on the lattice, meaning it represents the same or fewer transitions. And likewise, we may say that $i_1$ is ``stronger'' than $i_2$, or that $i_2$ is ``weaker'' than $i_1$. The lattice's top element $\itop$ represents the set of all state transitions and the bottom element $\ibot$ represents the identity relation:
\begin{align*}
	\itop\sdef\lambda v\in\VP:\dtop
	\quad&&\quad
	\ibot\sdef\lambda v\in\VP:\dbot.
\end{align*}
The join and meet operators, $\ijoin$ and $\imeet$, are defined component-wise:
\begin{align*}
	i_1\ijoin i_2\sdef\lambda v\in\VP: i_1(v)\djoin i_2(v)
	\quad&&\quad
	i_1\imeet i_2\sdef\lambda v\in\VP: i_1(v)\dmeet i_2(v).
\end{align*}

\subsection{Small Example}
\label{sec:example}

\noindent
In this section, we provide a small example to demonstrate the mechanics of our analysis. The following program is comprised of threads $T_0$ and $T_1$, global variables $x$ and $z$, and a variable $r$ that is local to $T_0$. We aim to prove that $r=0$ holds upon termination of the program, with respect to precondition $\var{true}$.

\vspace{-0.5\abovedisplayskip}{\centering\begin{varwidth}[t]{\textwidth}\begin{lstlisting}
T0:
1: r := 0
2: if z == 0:
3:     x := 0
4:     r := x
\end{lstlisting}\end{varwidth}\lstingSpace
\begin{varwidth}[t]{\textwidth}\begin{lstlisting}
T1:
1: if z == 1:
2:     x := 1
\end{lstlisting}\end{varwidth}\\}\vspace{\belowdisplayskip}

\noindent With a quick inspection of the program, we may infer that our postcondition is satisfied if $x=0$ holds before the assignment at line~4 in $T_0$. Looking at $T_1$ now, we see that $x$ may indeed be updated to~$1$, but crucially, only when $z=1$. This constraint on the conditions under which $x$ can be updated is the property that our analysis will infer in order to verify the program.

Our analysis is \emph{thread-modular}. For each thread $t$, we derive an assertion over the program state at each instruction, which we call its \emph{proof outline}. Our goal is for each of these assertions to over-approximate the states that are reachable at that location under any thread interleaving. Like many other thread-modular techniques~\cite{monat17,gupta11,ezudheen18,flanagan03,le20,carre09}, we accomplish this by generating a \emph{guarantee condition} for each thread in $t$'s environment, which is an over-approximation of its set of reachable state transitions. In our analysis, these guarantee conditions are represented as elements of the conditional-writes domain. By taking their join, we get a sound \emph{rely condition} for $t$, which over-approximates the set of state transitions that may be executed in $t$'s environment. To ensure that a given assertion in $t$'s proof outline holds in all thread interleavings, we weaken it so as to include the post-states of any number of transitions allowed by this rely. We now demonstrate this process on the above example.

To initialise our analysis, we must provide the following arguments:
\begin{enumerate}
	\item A state domain with lattice $\D$. We choose the \emph{constant domain}, which is a partial mapping from variables to their constant values.
	\item The variable sets $\V_0$ and $\V_1$ over which we would like the rely conditions of $T_0$ and $T_1$ to be defined respectively. We choose $\V_0=\V_1=\VP$, thus effectively treating all program variables as global.
	\item An integer $N$ specifying the precision of the $\stabilise$ function, where $0\le N\le|\V|$. In our case, $|\V|=3$ and we choose $N=3$ for maximum precision. We describe the precise meaning of this number in Section~\ref{sec:stabilise}.
\end{enumerate}
The analysis begins by deriving the reachable states for $T_0$ using our constant domain. We do not account for environment interference in this step. Elements of $\D$ are represented by mappings from variables to values; any variables not included are not constrained. For example, $[r\mapsto0]$ denotes the set of states where $r=0$ and $x$ and $z$ can have any value. 

\vspace{-0.5\abovedisplayskip}{\centering\begin{varwidth}[t]{\textwidth}\begin{lstlisting}
T0:
   $\dtop$
1: r := 0
   $[r\mapsto0]$
2: if z == 0:
       $[r\mapsto0,\:z\mapsto0]$
3:     x := 0
       $[r\mapsto0,\:z\mapsto0,\:x\mapsto0]$
4:     r := x
       $[r\mapsto0,\:z\mapsto0,\:x\mapsto0]$
\end{lstlisting}\end{varwidth}\\}\vspace{\belowdisplayskip}

\noindent We now derive a guarantee condition that captures the states under which each variable may be written to in $T_0$. By inspecting the assertions in our proof outline, we observe that the assignment to $r$ follows the assertion $\dtop$ and thus $r$ may be written to by $T_0$ while in any state. This fact is captured by the $\transitions$ function:
\[
	\transitions(\dtop,\:r:=0)=\ibot[r\mapsto\dtop]
\]
where $f[x\mapsto y]$ denotes the function $f$ with $x$ mapped to $y$. For example, the above is equal to:
\[
	[r\mapsto\dtop,\:x\mapsto\dbot,\:z\mapsto\dbot].
\]
Likewise, we observe that $x$ may only be written to when $r=0\land z=0$, and that $z$ is never written to. By applying the $\transitions$ function to each precondition-assignment pair and combining the results via $\ijoin$, we derive a guarantee condition $G_0\in\I$ which maps each variable to an over-approximation of the states (as elements of $\D$) under which they can be written to:
\begin{align*}
	G_0=[
	r\mapsto\dtop,\:
	x\mapsto[r\mapsto0,\:z\mapsto0],\:
	z\mapsto\dbot].
\end{align*}
We now convert this guarantee condition into a rely condition for $T_1$. We first map each variable not in $\V_1$ to $\dtop$ and havoc them in each write-condition with $\havocSymb$, effectively removing them from the rely. Since $\V_1=\VP$, there is nothing to do here. To weaken our rely to capture any number of environment transitions, our next step is to derive its transitive closure. The analysis recognises that transitivity is not satisfied due to the following two-step transition allowed under $G_0$:
\begin{center}
\begin{tabular}{ |c|c|c| } 
	\hline
	$r$ & $x$ & $z$ \\
	\hline
	1 & 0 & 0 \\
	\hline
\end{tabular}
$\to$
\begin{tabular}{ |c|c|c| } 
	\hline
	$r$ & $x$ & $z$ \\
	\hline
	0 & 0 & 0 \\
	\hline
\end{tabular}
$\to$
\begin{tabular}{ |c|c|c| } 
	\hline
	$r$ & $x$ & $z$ \\
	\hline
	0 & 1 & 0 \\
	\hline
\end{tabular}
\end{center}
The first transition updates $r$ and is permitted due to $r$ being mapped to $\dtop$. The second transition updates $x$ and is permitted since the second state is captured by the write-condition $[r\mapsto0,\:z\mapsto0]$ of $x$. However, the transition from the first to the third state is not captured by our rely, since it involves an update to $x$, and the first state does not satisfy the write-condition of $x$ since $r=1$. To remedy this, our analysis weakens the write-condition of $x$ using the $close$ function:
\begin{align*}
	R_1=\close(G_0)=[
	r\mapsto\dtop,\:
	x\mapsto[z\mapsto0],\:
	z\mapsto\dbot].
\end{align*}
We now derive the reachable states for $T_1$. However, this time we must make sure to ``stabilise'' each generated intermediary assertion under $R_1$, which ensures our derived reachable states account for all possible interference from $T_1$'s environment. Since $R_1$ is transitive, an assertion $d_1\in\D$ can be stabilised by deriving some $d_2\in\D$ that captures all states reachable within just one step captured by $R_1$ from any pre-state in $d_1$, and taking $d_1\djoin d_2$. This is achieved by applying the $\stabilise$ function, where $\stabilise(R_1,d_1)$ is a $\D$ element greater than $d_1$ that accounts for all one-step transitions captured by $R_1$. Our analysis performs this operation with a degree of precision specified by the parameter $N$ that was provided upon initialisation (see Section~\ref{sec:stabilise}). We derive the following reachable states which, in our small example, are already stable upon generation:

\vspace{-0.5\abovedisplayskip}{\centering\begin{varwidth}[t]{\textwidth}\begin{lstlisting}
T1:
   $\dtop$
1: if z == 1:
       $[z\mapsto1]$
2:     x := 1
       $[z\mapsto1,\:x\mapsto1]$

\end{lstlisting}\end{varwidth}\\}\vspace{\belowdisplayskip}

\noindent We now generate a guarantee $G_1$ for $T_1$:
\begin{align*}
	G_1=[
	r\mapsto\dbot,\:
	x\mapsto[z\mapsto1],\:
	z\mapsto\dbot].
\end{align*}
Notice this guarantee captures our desired constraint: that $x$ is only updated when $z=1$. This guarantee is transitive as-is and since $\V_0=\VP$ we can immediately derive $R_0$ as $G_1$. We now reconstruct (from scratch) the proof outline for $T_0$ under this new rely. In our simple example, the stabilised assertions are identical to the last iteration and we have thus reached a fixpoint where we can return our current rely-guarantee conditions and proof outlines. The conjunction of the derived postconditions for $T_0$ and $T_1$ suffice to establish that $r=0$ holds upon termination of the program.

The following subsections introduce the $\stabilise$, $\transitions$, and $\close$ functions that we use to respectively stabilise assertions under rely conditions, derive guarantee conditions from proof outlines, and derive transitive closures for interference domain elements.

\subsection{The {\normalfont\textit{stabilise}} Function}
\label{sec:stabilise}

In this section, we introduce the interference domain function:
\[
	\stabilise:\I\times\D\to\D
\]
and then define it for our conditional-writes domain. The purpose of this function is to capture the set of all post-states that are reachable within \emph{one} transition of some $i\in\I$ (representing a rely) from any pre-state captured by some $d\in\D$:
\[
	\dgamma(\stabilise(i,d))\supseteq\dgamma(d)\cup\{s_2\in\s\mid\exists s_1\in\dgamma(d):(s_1,s_2)\in\igamma(i)\}.\tag{1}
\]
When $i$ is transitive, this is enough to ensure that $\stabilise(i,d)$ captures all states reachable within any number of transitions in $i$ from an initial state in $d$, thereby fully accounting for interference from an environment represented by $i$. The fact that $\stabilise(i,d)$ also captures $d$ accounts for the fact that the environment may not take any step, and thus accommodates non-reflexive interference domains which do not implicitly capture $\{(s_1,s_2)\in\s\times\s\mid s_1=s_2\}$ as the conditional-writes domain does.

However, note that (1) is an over-approximation and hence there may exist a transition $(s_1,s_2)\in\igamma(i)$ where $s_1\notin\dgamma(d)\land s_1\in\dgamma(\stabilise(i,d))$ but $s_2\notin\dgamma(\stabilise(i,d))$. Consequently, although $\stabilise(i,d)$ soundly captures the effects of interference by $i$ on $d$, it may not be closed under transitions in $i$. The practical effect of this limitation on our analysis is that repeated applications of $\stabilise$ may yield increasingly weaker abstract states, and we may therefore expect the analysis to be more conservative for longer threads, where $\stabilise$ is invoked more often during the generation of proof outlines.

Most existing thread-modular analyses achieve this missing closure by repeatedly applying a rely condition until a fixpoint is reached. In our framework, this is equivalent to defining $\stabilise(i,d)$ to be a fixpoint over some function that satisfies (1). That is, we iteratively weaken $d$ until there are no longer any transitions in $i$ with pre-states in $d$ and post-states not in $d$. In addition to achieving this closure, this approach has the performance advantage that $i$ does not need to be transitive for the sound capture of interference, which avoids the need to derive transitive closures for rely conditions. This raises an interesting question that we have not yet seen discussed in the literature: Is it more expensive to derive a fixpoint for every generated intermediary assertion than it is to derive a transitive closure for every generated rely condition? We investigate this problem by implementing two modes of our analysis: one where rely conditions are made to be transitive and $\stabilise$ is simply defined as a function that satisfies (1), and another where rely conditions are not made to be transitive and $\stabilise$ is defined as a fixpoint over the definition in the first mode. From now on, we use the notation $f^*$ to denote the version of the function $f$ that is used in the second (non-transitive) mode.

Let us now define $\stabilise(i,d)$ in a way that satisfies (1). A naive definition is given by havocing in $d$ every variable with a write-condition that intersects with $d$:
\[
	\havoc{d}{\{v\in\VP\mid d\dmeet i(v)\ne\dbot\}}.
\]
Intuitively, if $d$ does not intersect $i(v)$ then the environment cannot update $v$ while the program is in any state captured by $d$. However, if $d$ does intersect $i(v)$ then the environment may update $v$ while in a state captured by $d$ and, by adjusting $d$ to $\havoc{d}{\{v\}}$, we capture all the possible post-states resulting from an arbitrary update to $v$ from any state captured by $d$.

This definition is sound (in that it satisfies (1)), but also very conservative in that we are failing to account for the fact that, when $v$ is updated in a state captured by $d$, we must also be in a state captured by $i(v)$. To improve precision, rather than havocing variables $v$ in $d$ we could havoc them in $d \dmeet i(v)$ as follows.
\[
	d\djoin\bigdjoin_{v\in\V}\havoc{(d\dmeet i(v))}{\{v\}}.
\]
However, this definition is unsound! We are failing to account for those transitions in $i$ that update multiple variables. To achieve soundness, we must capture the effects of transitions that update any combination $V\in\powset(\VP)$ of variables. To do so, we first capture the set of states under which all variables $v\in V$ can be updated within one transition of $i$ while in a state captured by $d$, which is the meet of $d$ and all their write-conditions: $d\dmeet\bigsqcap^\sharp_{v\in V}i(v)$. We then simulate the effect of variable updates by havocing $V$. Our new definition is thus:
\[
	d\djoin\bigdjoin_{V\in\powset(\VP)}\havoc{\left(d\dmeet\bigdmeet_{v\in V}i(v)\right)}{V}.\tag{2}
\]
However, iterating over the powerset of $\V$ for every call to $\stabilise$ is infeasibly inefficient, so we must now decide whether to make our analysis conservative and fast, or precise and slow. Rather than providing a fixed definition, we accept a parameter $N$ that allows the user to determine the appropriate trade-off. For all variable sets $V$ with cardinality $|V|\le N$, we consider the effects of transitions that update $V$ precisely, as in (2). For the other, larger variable sets, we use a special technique to over-approximate their effects that only requires that we iterate over variable sets with cardinality $N+1$.

To formalise our definition, let us first derive the set of states $X$ that may be reached by executing any transition in $i$ that updates $N$ or fewer variables. This ``precise'' part of our definition is derived from (2):
\[
	X=
	\bigdjoin_{V\in\powset(\V),\:|V|\le N}
	\havoc{
		\left(
			d\dmeet\bigdmeet_{v\in V}
			i(v)
		\right)
	}{
		V
	}.
\]
The corresponding ``conservative'' part computes a coarse over-approximation of the set of states that are reachable from $d$ via transitions in $i$ that update more than $N$ variables. The following technique can be used to derive this over-approximation efficiently by considering only those variable sets with cardinality $N + 1$, rather than all sets with cardinality $>N$. Let us first collect the set $\V_N$ of all of such sets $\{v_1,v_2,...,v_{N+1}\}$ where the meet over the write conditions of $v_1,v_2,...,v_{N+1}$ intersect with $d$. These are the variable sets with cardinality $N+1$ in which all variables may be updated within one step of $i$ from a state captured by $d$:
\[
	\V_N=\left\{V\in\powset(\V) \text{\LARGE $\,\mid\,$} |V|=N+1\land\left(d\dmeet\bigdmeet_{v\in V}i(v)\right)\ne\dbot\right\}.
\]
We now collect an over-approximation of the set of states in $d$ in which \emph{any} of these variable combinations in $\V_N$ may be written to:
\[
	\bigdjoin_{V\in\V_N}d\dmeet\bigdmeet_{v\in V}i(v).
\]
Finally, we havoc in the above all variables that appear in $\VP_N$:
\[
	Y=\havoc{
		\left(\ \ 
			\bigdjoin_{V\in\V_N}
			d\dmeet\bigdmeet_{v\in V}i(v)
		\right)
	}{
		\bigcup_{V\in\V_N}V.
	}
\]
In Appendix~\ref{app:stabilise}, we show that $Y$ over-approximates the post-states of not only the transitions that update those variable sets in $\V_N$, but every subset of $\VP$ with cardinality $>N$.

Our definition of $\stabilise_N(i,d)$ is the join of $d$ with the precise over-approx\-imation $X$ of post-states of transitions that update $N$ or fewer variables, and the less precise over-approximation $Y$ of post-states of transitions that update more than $N$ variables:
\[
	\stabilise_N(i,d)\sdef d\djoin X\djoin Y.
\]
We prove that our definition satisfies (1) in Appendix~\ref{app:stabilise}. In the non-transitive mode of our analysis, we define the stabilisation function as the least fixpoint of the above:
\[
	\func{stabilise}^*_N(i,d)\sdef\mu d:\stabilise_N(i,d).
\]

As indicated above, the computation of $\stabilise_N(i,d)$ may be slow for a large $N$ due to the requirement to iterate over a large portion of $\powset(\VP)$. In practice, we can implement a simple optimisation to reduce this domain. For any variable set $V_0\in\powset(\VP)$, if the meet over the write-conditions of the variables in $V_0$ is $\dbot$, we can ignore any superset of $V_0$. This is due to the fact that, since $V$ is a superset of $V_0$, the meet over the write-conditions of the variables in $V$ is also $\dbot$. This optimisation is formalised in Appendix~\ref{app:stabilise-opt}.

\subsection{The {\normalfont\textit{transitions}} Function}

In addition to the $\stabilise$ function, our analysis requires a method for generating guarantee conditions from proof outlines. Specifically, we require a definition of the function:
\[
	\transitions:\D\times\assign\to\I
\]
Intuitively, $\transitions(d,A)$ captures all state transitions that may occur by executing the (local or global) assignment $A$ under any state captured by $d$. Formally:
\[
	\igamma(\transitions(d,A))\supseteq\left\{(s_1,s_2)\in\s\times\s\mid
	s_1\in\dgamma(d)\land s_2=\sconc(A,s_1)
	\right\}.
\]
In our approach, we simply map each assigned variable to $d$ and map the others to $\dbot$:
\[
	\transitions(d,\langle x_1,...,x_n:=e_1,...,e_n\rangle)\sdef
	\ibot[v\in\{x_1,...,x_n\}\mapsto d].
\]
where $f[x\in X\mapsto y]$ denotes the function $f$ with each $x\in X$ mapped to $y$.

A guarantee for the whole thread can be derived by taking the join of $\transitions(d,A)$ over each precondition-assignment pair $(d,A)$ in its proof outline. We prove the soundness of our definition in Appendix~\ref{app:transitions}.

\subsection{The {\normalfont\textit{close}} Function}
\label{sec:transitive_closure}

Our analysis requires rely conditions to be transitive. In Appendix~\ref{app:transitivity}, we prove that this property is violated if and only if $i \in \I$ captures a transition $(s_1,s_2)$ from a state $s_1\notin\dgamma(i(v))$ to a state $s_2\in\dgamma(i(v))$ for some $v\in\VP$. In this case, $i$ contains a transition $(s_2,s_3)$ that modifies $v$, but cannot contain $(s_1,s_3)$ since $s_1$ is not in $v$'s write-condition. Therefore, $i$ is transitive if and only if, for all $v\in\VP$:
\[
	\forall(s_1,s_2)\in\igamma(i):s_1\notin i(v)\implies s_2\notin i(v).\tag{3}
\]
We therefore define a function:
\[
	\close:\I\to\I
\]
that satisfies (3) for all $v\in\VP$.

To preserve the property that our rely captures all transitions of the other threads' guarantee conditions, we also require $i\isqsubseteq \close(i)$, i.e., in deriving $close(i)$ we are only able to \emph{weaken} (that is, add states to) the write-conditions that variables are mapped to in $i$.

Now, let us imagine that $i$ is not transitive and it therefore captures a transition $(s_1,s_2)$ from a state $s_1\notin\dgamma(i(v))$ to a state $s_2\in\dgamma(i(v))$ for some $v\in\VP$. The only means of rectifying such a counterexample via weakening a write-condition is to add $s_1$ to $\dgamma(i(v))$. A solution for deriving a transitive closure for $i$ thus depends on a method for deriving, for some $v\in\VP$, an over-approximation of all states $s_1\notin \dgamma(i(v))$ where there exists a transition $(s_1,s_2)\in i$ for some $s_2\in \dgamma(i(v))$.

A very coarse approach may be to simply collect all the states captured by the write-conditions of the variables that appear in $i(v)$. However, since we only invoke the transitive closure procedure once per derivation of a rely condition, and (especially due to the abstractions provided by $\D$ and $\I$) we do not expect to perform many of these derivations during our analysis, it may be useful to investigate more precise methods.

The intuition for our method is as follows: Suppose there exists a violating transition $(s_1,s_2)$ as described above, that only modifies variable $v'$. Note that $v'$ must be a different variable to $v$, since $s_1\notin\dgamma(i(v))$ and thus $v$ cannot be modified in the state $s_1$. Since our transition only modifies $v'$, we must have that $s_1$ belongs to the set of states obtained by havocing $v'$ in $s_2$:
\[
	s_1\in\{s\in\s\mid\forall u\in\VP\setminus\{v'\}:s(u)=s_2(u)\}
\]
Since $s_2\in \dgamma(i(v))$, we can obtain an over-approximation of this set with $\havoc{i(v)}{v'}$. Moreover, it should be obvious that $s_1$ is captured by $i(v')$, by virtue of $i$ capturing $(s_1,s_2)$ where $s_1(v')\ne s_2(v')$. Combining these two facts, we have:
\[
	s_1\in\dgamma(\havoc{i(v)}{\{v'\}})\cap\dgamma(i(v')).
\]
In our analysis, we can derive an over-approximation of this set as a $\D$ element:
\[
	\havoc{i(v)}{\{v'\}}\dmeet i(v').
\]
Intuitively, this represents the subset of states in $i(v')$ that may result in a post-state in $i(v)$ when $v'$ is updated. To account for the update of any other variable, we weaken $i(v)$ by taking the join of the above for every $v'\in\VP\setminus\{v\}$:
\[
	i(v)\djoin\bigdjoin_{v'\in\VP\setminus\{v\}}\havoc{i(v)}{\{v'\}}\dmeet i(v').
\]
Since $\havoc{i(v)}{\{v\}}\dmeet i(v)= i(v)$, we may simplify this definition by adding $v$ to the domain of the join:
\[
	\bigdjoin_{v'\in\VP}\havoc{i(v)}{\{v'\}}\dmeet i(v').
\]
Now, recall that our reasoning above relied on $v'$ being the only variable updated, whereas our conditional-writes domain captures transitions that modify multiple variables. For a true over-approximation of violating pre-states, we must actually iterate through the entire powerset of $\VP$:
\[
	\bigdjoin_{V\in\powset(\VP)}\havoc{i(v)}{V}\dmeet\bigdmeet_{v'\in V}i(v').\tag{4}
\]
However, we discuss optimisations in Appendix~\ref{sec:optimisations} that significantly reduce this search space.

We now have a method for weakening $i(v)$ so as to capture the pre-states of all transitions in $i$ with post-states in $i(v)$. However, the new states added to $i(v)$ may themselves be post-states of other transitions captured by $i$, and thus we must derive a fixpoint on (4) to achieve (3) for $v$. Moreover, the new transitions added to $i$ by weakening $v$'s write-condition may invalidate (3) for the other variables, and thus to compute a transitive closure with this method, we must derive a fixpoint on (4) over \emph{all} variables:
\[
	\close(i)\sdef\mu i:\lambda v:\bigdjoin_{V\in\powset(\VP)}\havoc{i(v)}{V}\dmeet\bigdmeet_{v'\in V}i(v').
\]
Let $i'=\close(i)$. Upon termination of the fixpoint, we have:
\[
	\forall v\in\VP:i'(v)\dsqsupseteq\bigdjoin_{V\in\powset(\VP)}
	\havoc{i'(v)}{V}\dmeet\bigdmeet_{v'\in V} i'(v')
\]
which we prove to be sufficient for transitivity in Appendix~\ref{app:transitivity}.

\section{Rely-Guarantee Generation}
\label{sec:analysis}

A complete summary of the data structures and functions defined in the condi\-tional-writes domain is provided in Appendix~\ref{app:summary}. The purpose of this section is to formalise the process described in Section~\ref{sec:example}, thereby providing an example of integrating our domain into a complete program analysis, similar to the style of~\cite{mine17}.

Our analysis involves iteratively building up rely conditions $R:\T\to\I$ and guarantee conditions $G:\T\to\I$ for each thread. Our rely conditions are directly computed from the generated guarantee conditions using the $\rely$ function:
\[
	\rely:\T\times(\T\to\I)\to\I
\]
which, in the transitive mode of our analysis, is defined as the transitive closure of the join over all environment guarantees:
\[
	\rely(t,G)\sdef\func{close}\circ\func{reduce}_t\left(\ \ \bigijoin_{t'\in\T\setminus\{t\}}G(t')\right)
\]
where $\func{reduce}_t(i)$ weakens $i$ to eliminate all variables not in $\V_t$:
\[
	\func{reduce}_t(i)\sdef i\left[v\in\V\setminus\V_t\mapsto\dtop\right][v\in\V_t\mapsto\havoc{i(v)}{\V\setminus\V_t}].
\]
When running in the non-transitive mode described in Section~\ref{sec:stabilise}, we simply skip the transitive closure step:
\[
	\rely^*(t,G)\sdef\func{reduce}_t\left(\ \ \bigijoin_{t'\in\T\setminus\{t\}}G(t')\right).
\]

To derive a guarantee condition for a thread $t\in\T$, we generate a thread-local proof outline of $t$ using a state domain with lattice $\D$, with intermediary assertions weakened to account for interference specified by $t$'s rely. As reachable states are discovered, we use the $\transitions$ function to derive the state transitions they induce, and collect these transitions into a new guarantee for $t$. This analysis is formalised with the semantics $\mathbb{S}$ defined in Figure~\ref{fig:semantics}, where $\mathbb{S}(c,\:(d,\:r,\:g))$ is a tuple of which the third element is a sound guarantee condition of program $c$ with respect to precondition $d$ and rely $r$. To provide a concise formalisation, we define a join operator $\sqcup$ on tuples of type $\D\times\I\times\I$:
\[
	\sqcup:(\D\times\I\times\I)\times(\D\times\I\times\I)\to (\D\times\I\times\I)
\]
\[
	(d_1,\:r_1,\:g_1)\sqcup(d_2,\:r_2,\:g_2)\sdef\left(d_1\djoin d_2,\:r_1\ijoin r_2,\:g_1\ijoin g_2\right).
\]

\begin{figure}
\[
	\mathbb{S}:\instr\times(\D\times\I\times\I)\to\D\times\I\times\I
\]
\vspace{-5mm}
\begin{align*}
	&\mathbb{S}(c_1;c_2,\:X)\sdef\mathbb{S}(c_2,\:\mathbb{S}(c_1,\:X))
\end{align*}
\vspace{-5mm}
\begin{align*}
	\mathbb{S}(c:\assign,\:(d,\:r,\:g))\sdef&\text{let }d'=\stabilise_N(r,\:d)\text{ in}\\
	&(\sd(c,\:d'),\:r,\:g\ijoin\transitions(d',\:c))
\end{align*}
\vspace{-5mm}
\begin{align*}
	\mathbb{S}(\text{if }b\text{ then }c_1\text{ else }c_2,\:X)\sdef
	&\text{let }X_1=\mathbb{S}(c_1,\:\mathbb{B}(b,\:X))\text{ in }\\
	&\text{let }X_2=\mathbb{S}(c_2,\:\mathbb{B}(\neg b,\:X))\text{ in }\\
	&X_1\sqcup X_2
\end{align*}
\vspace{-5mm}
\begin{align*}
	&\mathbb{S}(\text{while }b\text{ do }c,\:X) \sdef
	\mathbb{B}(\neg b,\:\mu X\:.\:X\sqcup\mathbb{S}(c,\:\mathbb{B}(b,\:X))
\end{align*}
\vspace{-5mm}
\begin{align*}
	&\mathbb{B}(b,\:(d,\:r,\:g))\sdef(\bd(b,\:\stabilise_N(r,\:d)),\:r,\:g)
\end{align*}
\caption{Collective semantics for generating guarantee conditions.}
\label{fig:semantics}
\end{figure}

We define our entire analysis by the function $\func{Analyse}$, which takes a program precondition and returns a set of rely and guarantee conditions for the program:
\[
	\func{Analyse}:\D\to(\T\mapsto\I)\times(\T\mapsto\I).
\]
It is defined as:
\begin{align*}
	\func{Analyse}(\var{pre})\sdef&\mu R,G:\\
	&~~~(\lambda t\in\T:\rely(t,\:G),\ \,
	\lambda t\in\T:\mathbb{S}(t,\:(\var{pre},\:\rely(t,\:G),\:\ibot)).3)
\end{align*}
where $X.3$ refers to the third element of tuple $X$. In the non-transitive mode of our analysis, the functions $\rely^*$ and $\stabilise^*_N$ are used in place of $\rely$ and $\stabilise_N$ respectively.

\section{Evaluation}
\label{sec:eval}

We implemented both modes of our analysis in a simple Python prototype\footnote{Available at https://doi.org/10.5281/zenodo.18589190.}, and evaluated each mode using two state domains: the constant domain, and the powerset completion of the constant domain~\cite{mine17}. For all tests, we choose $N=|\V|$ and $\V_t=\V$ for each thread to attain maximum precision. We evaluated our analysis on five concurrent programs that are verifiable with the conditional-writes domain, and record our results in Table~\ref{tab:verification_results}. All experiments were conducted on a 2019 MacBook Pro with a 2.3 GHz Intel Core i5 processor.

\begin{table}[h]
\centering
\small\vspace{-4mm}
\begin{tabular}{|l|ccc|ccc|ccc|ccc|}
\hline
\multirow{3}{*}{Program} & \multicolumn{6}{c|}{Constant Domain} & \multicolumn{6}{c|}{Powerset Completion of Const. Dom.} \\
\cline{2-13}
& \multicolumn{3}{c|}{Non-Transitive} & \multicolumn{3}{c|}{Transitive} & \multicolumn{3}{c|}{Non-Transitive} & \multicolumn{3}{c|}{Transitive} \\
\cline{2-13}
& Ver. & Ops & Time (s) & Ver. & Ops & Time (s) & Ver. & Ops & Time (s) & Ver. & Ops & Time (s) \\
\hline
reset & $\times$ & 123 & 0.0004 & $\times$ & 168 & 0.0005 & \checkmark & 533 & 0.0020 & $\times$ & 786 & 0.0067 \\
circular & \checkmark & 620 & 0.0014 & \checkmark & 836 & 0.0022 & \checkmark & 1019 & 0.0053 & \checkmark & 1235 & 0.0064 \\
mutex1 & \checkmark & 67 & 0.0004 & \checkmark & 105 & 0.0006 & \checkmark & 237 & 0.0013 & \checkmark & 342 & 0.0021 \\
mutex2 & $\times$ & 48 & 0.0006 & $\times$ & 51 & 0.0003 & \checkmark & 318 & 0.0018 & \checkmark & 339 & 0.0016 \\
spinlock & $\times$ & 336 & 0.0012 & $\times$ & 291 & 0.0009 & \checkmark & 1291 & 0.0115 & \checkmark & 2355 & 0.0116 \\
\hline
\end{tabular}\vspace{2mm}
\caption{Verification results across different state domains and analysis modes.}\vspace{-8mm}
\label{tab:verification_results}
\end{table}

The program $\var{spinlock}$ can be found in  \cite{winter21}. All other programs were crafted by the authors and can be found alongside the prototype implementation. We record whether the program was verified by our framework under \emph{Ver.}, the number of $\djoin$ and $\dmeet$ operations under \emph{Ops}, and the performance time under \emph{Time}.

Our preliminary results show that the transitive mode is nearly always slower and requires more state lattice operations than the non-transitive mode. Moreover, the transitive mode is less precise in that it fails to verify $\var{reset}$ with the powerset completion domain, whereas the non-transitive mode succeeds. This provides a degree of justification for the widely adopted fixpoint approach to stabilisation. However, we note that most of our benchmark programs are small algorithms, and these results may not generalise to programs with longer thread bodies, where the ratio of $\stabilise$ to $\rely$ calls is higher.

\section{Related Work}
\label{sec:rel}

Our work is inspired by the analyses introduced by Miné~\cite{mine12-older,mine12,mine14} and Monat and Min\'e~\cite{monat17} for generating rely-guarantee conditions with abstract interpretation. Of these, the framework presented in \cite{monat17} is the most similar to ours in that it is defined over an abstract domain for representing states, another abstract domain for representing interferences, and their own versions of the $\stabilise$ and $\transitions$ functions for converting between the two. However, their analysis requires that interferences are represented by a standard numeric abstract domain over primed and unprimed variables, where primed variables represent their values in the post-state of the interference. This allows one to use existing abstract domains to represent interferences, and also allows for partially-defined $\stabilise$ and $\transitions$ functions that only require the user to define functions for adding and removing certain variables (for example, the primed variables) from the chosen abstract domains. However, we found that the conditional-writes domain cannot be neatly expressed as a numeric domain over primed and unprimed variables. One reason for this is that it imposes different structural constraints on the primed and the unprimed variables, such as the fact that a primed variable can only be constrained to be equal to its unprimed counterpart (when the write-condition for the variable is false) and is otherwise unconstrained. Moreover, our framework allows users to fully define the $\stabilise$ and $\transitions$ functions, which allows for more complex definitions such as that of our $\stabilise$ function. Finally, all of the analyses in the aforementioned works perform stabilisation by deriving a fixpoint over the rely condition, whereas we provide an additional analysis mode where rely conditions are made to be transitive and interference can therefore be accounted for by over-approximating just one step of the rely. 

Our analysis is underpinned by the rely-guarantee proof method \cite{jones83,xu97}, which is known to be better suited to some classes of programs than other logics such as Concurrent Separation Logic~\cite{ohearn07,feng07}. Many techniques for generating rely-guarantee conditions do not use abstract interpretation and therefore do not enjoy the efficiency-precision trade-off afforded by an abstract domain~\cite{gupta11,ezudheen18,flanagan03,le20}. To the authors' knowledge, the existing techniques that do are restricted to particular abstractions of interferences such as non-relational domains that map variables to sets of constant post-state values~\cite{carre09,mine12-older}, or numeric domains over primed and unprimed variables~\cite{monat17}. In contrast, we abstract interferences into \emph{interference domains} that can represent arbitrary structures of rely-guarantee conditions, and allow these to be paired with any \emph{state domain} that implements the functions used in the definitions of $\stabilise$ and $\transitions$.

\section{Conclusion}
\label{sec:con}

In this work, we have introduced the notion of an \emph{interference domain} for abstracting the set of state transitions that may be executed by a thread or its environment. An interference domain is parameterised by a compatible state domain, and is defined by a lattice over sets of pairs of states representing transitions, as well as the functions $\stabilise$ and $\transitions$ for respectively applying and deriving these transitions. We have then introduced a simple example of an interference domain, called the conditional-writes domain, for capturing rely-guarantee conditions that constrain only the conditions under which variables are updated, whilst disregarding what they are updated to. We have then described a simple analysis framework that demonstrates how interference domains can be used to generate sound rely-guarantee conditions for concurrent programs.

A significant limitation of this framework is the overly simple target language. In future, we would like to extend our framework to support common language features like procedure calls, dynamic thread creation, and mutexes. We would also like to create a suite of interference domains that are more suitable for sophisticated programs with complex data structures such as arrays and pointers. To this end, we intend to adopt an example-driven approach where we will identify useful interference domains by inspecting concurrent programs used in practice. Finally, we intend to explore existing static analysis tools in which to implement our framework, to improve performance and enable better comparison with existing analyses.

\begin{credits}
 \subsubsection{\ackname}
 We would like to thank Kirsten Winter for her insightful comments and discussions on this paper.
 
\subsubsection{\discintname}
The authors have no competing interests to declare that are relevant to the content of this article.
\end{credits}

\bibliographystyle{splncs04}
\bibliography{main}

\newpage
\appendix

\section{Proofs}

\subsection{Soundness of {\normalfont\textit{stabilise}}}
\label{app:stabilise}

\begin{lemma}
For any $i\in\I$ and $d\in\D$ we have:
\[
	\dgamma(\stabilise_N(i,d))\supseteq\dgamma(d)\cup\{s_2\in\s\mid\exists s_1\in\dgamma(d):(s_1,s_2)\in\igamma(i)\}.
\]
\end{lemma}
\begin{proof}
For readability, let us first define the meet over all write conditions in some $i\in\I$ of a variable set $V$ as $M_i(V)$:
\[
	M_i:\powset(\VP)\to\D
\]
\[
	M_i(V)\sdef\bigdmeet_{v\in V}i(v).
\]
The definition of $\stabilise_N$ is thus:
\begin{align*}
	\stabilise_N&(i,d)\sdef\\&\text{let }X=\bigdjoin_{V\in\powset(\V),\:|V|\le N}
	\havoc{
		\left(
			d\dmeet M_i(V)
		\right)
	}{
		V
	}\text{ in}\\
	&\text{let }\V_N=\left\{V\in\powset(\V)\text{\LARGE{$\,\mid\,$}} |V|=N+1\land d\dmeet M_i(V)\ne\dbot\right\}\text{ in}\\
	&\text{let }Y=\havoc{
		\left(\ \ 
			\bigdjoin_{V\in\V_N}
			d\dmeet M_i(V)
		\right)
	}{
		\bigcup_{V\in\V_N}V
	}\text{ in}\\
	&d\djoin X\djoin Y.
\end{align*}
The definition is a composition of a ``precise'' component $X$ that captures transitions modifying $N$ or fewer variables, and a ``conservative'' component $Y$ that captures transitions modifying more than $N$ variables. As a first step, we prove the conservativeness of $Y$. Formally:
\[
	\bigdjoin_{V\in\powset(\V),\:|V|>N}
	\havoc{
		\left(
			d\dmeet M_i(V)
		\right)
	}{
		V
	}
	\dsqsubseteq
	Y.
\]
Since for any $d_1,d_2,d_3\in\D$ we have $d_1\dsqsubseteq d_3\land d_2\dsqsubseteq d_3\implies d_1\djoin d_2\dsqsubseteq d_3$, it suffices to prove that, for an arbitrary $V'\in\powset(\V)$ where $|V'|>N$:
\[
	\havoc{
		\left(
			d\dmeet M_i(V')
		\right)
	}{
		V'
	}
	\dsqsubseteq
	Y.
\]
Recall our constraint on $\havocSymb$:
\[
	\forall V\in\powset(\VP):\havoc{\dbot}{V}=\dbot.
\]
Thus if $d\dmeet M_i(V')=\dbot$ then the above reduces to $\dbot\dsqsubseteq Y$ which holds. Now, suppose $d\dmeet M_i(V')\ne\dbot$. Expanding $Y$, our proof obligation becomes:
\[
	\havoc{
		\left(
			d\dmeet M_i(V')
		\right)
	}{
		V'
	}\dsqsubseteq
	\havoc{
		\left(\ \ 
			\bigdjoin_{V\in\V_N}
			d\dmeet M_i(V)
		\right)
	}{
		\bigcup_{V\in\V_N}V
	}.\tag{1}
\]
Since $|V'|\ge N+1$, it is the union of all its $(N+1)$-subsets $V'_{\var{sub}}$, defined formally:
\[
	V'_{\var{sub}}=\{V\in\powset(\V)\mid V\subseteq V'\land|V|=N+1\}.
\]
Moreover, for any $V\in V'_{\var{sub}}$, we have $M_i(V')\dsqsubseteq M_i(V)$ and so $d\dmeet M_i(V)\ne\dbot$. Thus, by definition of $\V_N$:
\[
	V'_{\var{sub}}\subseteq\V_N.\tag{2}
\]
It follows that:
\[
	V'=\bigcup_{V\in V'_{\var{sub}}}V\subseteq\bigcup_{V\in\V_N}V.\tag{3}
\]
Let $V^-$ be a variable set in $V'_{\var{sub}}$. Then, since $V^-$ is a subset of $V'$, we have $M_i(V')\dsqsubseteq M_i(V^-)$, and thus it follows from (2) that:
\[
	d\dmeet M_i(V')\dsqsubseteq d\dmeet M_i(V^-)\dsqsubseteq\bigdjoin_{V\in\V_N}d\dmeet M_i(V).\tag{4}
\]
Recall that $\havocSymb$ is monotonic in that:
\[
	\forall V_1,V_2\in\powset(\V),d_1,d_2\in\D:V_1\subseteq V_2\land d_1\dsqsubseteq d_2\implies\havoc{d_1}{V_1}\dsqsubseteq\havoc{d_2}{V_2}.
\]
Thus from (3) and (4) we can deduce (1). Therefore $Y$ is indeed conservative and we can reduce our definition accordingly:
\begin{align*}
	&\stabilise_N(i,d)\\
	&=d\djoin X\djoin Y\\
	&=d\djoin X\djoin
	\havoc{
		\left(\ \ 
			\bigdjoin_{V\in\V_N}
			d\dmeet M_i(V)
		\right)
	}{
		\bigcup_{V\in\V_N}V
	}\\
	&\dsqsupseteq d\djoin X\djoin
	\bigdjoin_{V\in\powset(\V),\:|V|>N}
	\havoc{
		\left(
			d\dmeet M_i(V)
		\right)
	}{
		V
	}\\
	&=d\djoin \left(\ \ \bigdjoin_{V\in\powset(\V),\:|V|\le N}
	\havoc{
		\left(
			d\dmeet M_i(V)
		\right)
	}{
		V
	}\right)\\
	&\phantom{=d}\:\mkern1mu\djoin
	\left(\ \ \bigdjoin_{V\in\powset(\V),\:|V|>N}
	\havoc{
		\left(
			d\dmeet M_i(V)
		\right)
	}{
		V
	}\right)\\
	&=d\djoin \bigdjoin_{V\in\powset(\V)}
	\havoc{
		\left(
			d\dmeet M_i(V)
		\right)
	}{
		V
	}.
\end{align*}
Thus we aim to prove:
\begin{align*}
	&\dgamma\!\!\left(\ \ d\djoin 
	\bigdjoin_{V\in\powset(\V)}
	\havoc{
		\left(
			d\dmeet M_i(V)
		\right)
	}{
		V
	}
	\right)
	\\&\supseteq
	\dgamma(d)\cup\{s_2\in\s\mid\exists s_1\in\dgamma(d):(s_1,s_2)\in\igamma(i)\}.\tag{5}
\end{align*}
Case 1: Let $s\in\dgamma(d)$. Then for any $d'\in\D$ we have $s\in\dgamma(d\djoin d')$ and thus $s$ belongs to the left-hand side of (5).
Case 2: Let $s_2\in\s$ where $\exists s_1\in\dgamma(d):(s_1,s_2)\in\igamma(i)$. Then from the concretisation function of $\I$ we have:
\[
	\forall v\in\VP:s_2(v)\ne s_1(v)\implies s_1\in\dgamma(i(v)).
\]
Let $P$ be the maximal subset of $\VP$ where $\forall v\in P:s_2(v)\ne s_1(v)$. Then $s_1\in\dgamma(M_i(P))$, and since $s_1\in\dgamma(d)$ we have $s_1\in\dgamma(d\dmeet M_i(P))$. Since $P$ is the maximal subset as defined above, we have $\forall v\in\VP\setminus P:s_1(v)=s_2(v)$. Now recall our constraint on the `$\havocSymb$' operator:
\[
	\dgamma(\havoc{d}{V})\supseteq\left\{
		s\in\s\mid
		\exists s_0\in\dgamma(d):
			\forall v\in\VP\setminus V:s(v)=s_0(v)
	\right\}.
\]
Thus, taking $s_1$ for $s_0$ and $d\dmeet M_i(P)$ for $d$ in the above, we have:
\[
	s_2\in\dgamma(d\dmeet\havoc{M_i(P))}{P}.
\]
Therefore, since $P\in\powset(\VP)$ we have:
\[
	s_2\in\dgamma\!\!\left(\ \ 
	\bigdjoin_{V\in\powset(\V)}
	\havoc{
		\left(
			d\dmeet M_i(V)
		\right)
	}{
		V
	}
	\right)
\]
and thus $s_2$ belongs to the left hand side of (5).
\end{proof}

\subsection{Soundness of {\normalfont\textit{transitions}}}
\label{app:transitions}

\begin{lemma}
For any $d\in\D$ and assignment $A=\langle x_1,...,x_n:=e_1,...,e_n\rangle$ we have:
\[
	\igamma(\transitions(d,A))\supseteq\left\{(s_1,s_2)\in\s\times\s\mid
	s_1\in\dgamma(d)\land s_2=\sconc(A,s_1)
	\right\}.
\]
\end{lemma}
\begin{proof}
Recall the definition of $\transitions$:
\[
	\transitions(d,\langle x_1,...,x_n:=e_1,...,e_n\rangle)\sdef
	\ibot[v\in\{x_1,...,x_n\}\mapsto d]
\]
and the concretisation function of $\I$:
\[
	\igamma(i)\sdef\left\{
	\left(s,s'\right)\in\s\times\s\mid
	\forall v\in\VP:s'(v)\ne s(v)\implies s\in\dgamma(i(v))
	\right\}.
\]
It suffices to prove that for all $(s_1,s_2)\in\s\times\s$ where $s_1\in\dgamma(d)\land s_2=\sconc(A,s_1)$, we have $(s_1,s_2)\in\igamma(\transitions(d,A))$. That is:
\[
	\forall v\in\VP:s_2(v)\ne s_1(v)\implies s_1\in\dgamma(\ibot[v_0\in\{x_1,...,x_n\}\mapsto d](v)).\tag{1}
\]
We expect our concrete semantics to satisfy:
\[
	x_1'=e_1\land...\land x_n'=e_n\land\forall v\notin\{x_1,...,x_n\}:v'=v.
\]
Let $v\in\VP$ be some variable where $s_2(v)\ne s_1(v)$. Then $v\in\{x_1,...,x_n\}$ due to the above. Thus:
\[
	\ibot[v_0\in\{x_1,...,x_n\}\mapsto d](v)= d
\]
and thus the consequent in (1) reduces to $s_1\in\dgamma(d)$ which is already established.
\end{proof}

\subsection{Soundness of {\normalfont\textit{close}}}
\label{app:transitivity}

\begin{lemma}
Let some $i\in\I$ satisfy:
\[
	\forall v\in\VP:i(v)\dsqsupseteq\bigdjoin_{V\in\powset(\VP)}
	\havoc{i(v)}{V}\dmeet\bigdmeet_{v'\in V} i(v').\tag{1}
\]
Prove that $i$ satisfies the transitivity property:
\[
	\forall s_1,s_2,s_3\in\s:(s_1,s_2)\in\igamma(i)\land(s_2,s_3)\in\igamma(i)\implies
	(s_1,s_3)\in\igamma(i).
\]
\end{lemma}
\begin{proof}
Suppose $i$ is not transitive. Then there must exist states $s_1,s_2,s_3\in\s$ where:
\[
	(s_1,s_2)\in\igamma(i)\land(s_2,s_3)\in\igamma(i)\land
	(s_1,s_3)\notin\igamma(i).
\]
Recall the definition of $\igamma(i)$:
\[
	\igamma(i)\sdef\left\{
		\left(s,s'\right)\in\s\times\s\mid
		\forall v\in\VP:s'(v)\ne s(v)\implies s\in\dgamma(i(v))
	\right\}.
\]
By this definition, since $(s_1,s_3)\notin\igamma(i)$, there must be some variable, call $\hat v$, where:
\[
	s_3(\hat v)\ne s_1(\hat v)\land s_1\notin\dgamma(i(\hat v)).
\]
Since $s_1\notin\dgamma(i(\hat v))$, we have $s_1(\hat v)=s_2(\hat v)$ by definition of $\igamma(i)$. Thus, since $s_1(\hat v)\ne s_3(\hat v)$, we have $s_2(\hat v)\ne s_3(\hat v)$, and hence $s_2\in\dgamma(i(\hat v))$. Thus we have:
\[
	s_1\notin\dgamma(i(\hat v))\land s_2\in\dgamma(i(\hat v)).\tag{2}
\]
Intuitively, if $i$ is not transitive then it captures a transition from a state not in $\hat v$'s write-condition to a state that is. We will now show that this is impossible given (1). First, let $V^\Delta$ be the set of all variables that differ in $s_1$ and $s_2$. Then we have:
\[
	\forall v\in\VP\setminus V^\Delta:s_1(v)=s_2(v)\tag{3}
\]
\[
	s_1\in\bigcap_{v\in V^\Delta}\dgamma(i(v)).\tag{4}
\]
Recall our constraint on $\havocSymb$, that for any $d\in\D,V\in\powset(\V)$:
\[
	\dgamma(\havoc{d}{V})\supseteq\left\{
		s\in\s\mid
		\exists s_0\in\dgamma(d):
			\forall v\in\VP\setminus V:s(v)=s_0(v)
	\right\}.
\]
By substituting $i(\hat v)$ for $d$ and $V^\Delta$ for $V$, we have:
\[
	\dgamma\left(\havoc{i(\hat v)}{V^\Delta}\right)\supseteq
	\left\{s\in\s\mid
	\exists s_0\in\dgamma(i(\hat v)):
	\forall v\in\VP\setminus V^\Delta:s(v)=s_0(v)\right\}.
\]
Now, from (2) we have $s_2\in\dgamma(i(\hat v))$, and recall from (3) that $\forall v\in\VP\setminus V^\Delta:s_1(v)=s_2(v)$. Thus, $s_1$ is an element of the RHS set above, and thus:
\[
	s_1\in\dgamma\left(\havoc{i(\hat v)}{V^\Delta}\right)\tag{5}.
\]
From (4) and (5) we have:
\begin{align*}
	s_1\in\:&\dgamma\left(\havoc{i(\hat v)}{V^\Delta}\right)\cap\bigcap_{v\in V^\Delta}\dgamma(i(v))\\
	\subseteq\:&\bigcup_{V\in\powset(\VP)}\dgamma\left(\havoc{i(\hat v)}{V}\right)\cap\bigcap_{v\in V}\dgamma(i(v))\\
	\subseteq\:&\dgamma\!\!\left(\ \ \bigdjoin_{V\in\powset(\VP)}\havoc{i(\hat v)}{V}\dmeet\bigdmeet_{v\in V}i(v)\right)\\
	\subseteq\:&\dgamma(i(\hat v))\text{ from (1)}.
\end{align*}
Thus $s_1\in\dgamma(i(\hat v))$, which contradicts (2).
\end{proof}

\section{Optimisations}

\subsection{Optimisations for Computing {\normalfont\textit{stabilise}}}
\label{app:stabilise-opt}

As indicated in Section~\ref{sec:stabilise}, the computation of $\stabilise_N(i,d)$ may be slow for a large $N$ due to the requirement to iterate over a large portion of $\powset(\VP)$. In practice, we can implement a simple optimisation to reduce this domain. For any variable set $V_0\in\powset(\VP)$, if the meet over the write-conditions of the variables in $V_0$ is $\dbot$, we can ignore any superset of $V_0$. Formally, we can ignore any $V\in\powset(\VP)$ satisfying:
\[
	\exists V_0\in\powset(\VP):V_0\subseteq V\land\bigdmeet_{v\in V_0}i(v)=\dbot.
\]
This is due to the fact that, since $V$ is a superset of $V_0$, the meet over the write-conditions of the variables in $V$ is also $\dbot$:
\[
	\bigdmeet_{v\in V}i(v)=\dbot
\]
and thus:
\[
	\havoc{\left(d\dmeet\bigdmeet_{v\in V}i(v)\right)}{V}=\dbot.
\]
This optimisation can be best taken advantage of with an implementation that traverses $\powset(\VP)$ bottom-up. When we encounter some $V_0\in\powset(\VP)$ satisfying $\bigdmeet_{v\in V_0}i(v)=\dbot$, we can ignore all supersets of $V_0$.

\subsection{Optimisations for Computing {\normalfont\textit{close}}}
\label{sec:optimisations}

Our definition of $\close$ requires taking the join over $\powset(\VP)$. Although this is slow, we can implement two significant optimisations to significantly reduce this domain.

\subsubsection{Optimisation 1}

We only need to consider the powerset over variables constrained in $i(v)$. That is, we can remove from our domain any variable sets $V\in\powset(\VP)$ containing some $u\in V$ such that $\havoc{i(v)}{\{u\}}= i(v)$. This is due to the fact that, as noted in Section~\ref{sec:abs}:
\[
	\forall V_1,V_2\in\powset(\V),d\in\D:\havoc{d}{V_1}=d\implies\havoc{d}{V_1\cup V_2}=\havoc{d}{V_2}
\]
and thus, taking $d=i(v)$, $V_1=\{u\}$, and $V_2=V\setminus\{u\}$:
\[
	\havoc{i(v)}{V}=\havoc{i(v)}{(V\setminus\{u\})}
\]
and thus:
\[
	\havoc{i(v)}{V}\dmeet\bigdmeet_{v'\in V}i(v')\dsqsubseteq\havoc{i(v)}{(V\setminus\{u\})}\dmeet\bigdmeet_{v'\in V\setminus\{u\}}i(v').
\]
Since the states captured by domain element $V\setminus\{u\}$ are a superset of those captured by $V$, we can safely remove $V$ from our domain.

\subsubsection{Optimisation 2}

If during our traversal of $\powset(\VP)$ we encounter some $V_0\in\powset(\VP)$ where the result of the havoc operation fully encompasses the meet over the write-conditions of the variables in $V_0$, we can safely ignore any strict superset $V$ of $V_0$. Formally, we can ignore any variable set $V\in\powset(\VP)$ satisfying:
\[
	\exists V_0\subset V:\havoc{i(v)}{V_0}\dsqsupseteq\bigdmeet_{v'\in V_0}i(v').\tag{1}
\]
This is due to the fact that:
\[
	\havoc{i(v)}{V_0}\dmeet\bigdmeet_{v'\in V_0}i(v')=\bigdmeet_{v'\in V_0}i(v')
\]
and therefore, for the variable set $V_0$, it suffices to ignore the havoc operation and simply join the result of the meet operation. Now note that, since $V\supset V_0$, we have:
\[
	\bigdmeet_{v'\in V}i(v')\dsqsubseteq\bigdmeet_{v'\in V_0}i(v')
\]
and therefore:
\[
	\havoc{i(v)}{V}\dmeet\bigdmeet_{v'\in V}i(v')\dsqsubseteq\bigdmeet_{v'\in V_0}i(v').
\]
Since the states captured by domain element $V_0$ are a superset of those captured by $V$, we can safely remove $V$ from our domain. Note that the use of this optimisation requires a bottom-up traversal of $\powset(\VP)$.

\section{Formal Summary of The Conditional-Writes Domain}
\label{app:summary}

This section provides a complete formal summary of the conditional-writes domain. We define:
\begin{description}
	\item The concretisation function $\igamma:\I\to\powset(\s\times\s)$:
	\[
		\igamma(i)\sdef\left\{
		\left(s_1,s_2\right)\in\s\times\s\mid
		\forall v\in\VP:s_2(v)\ne s_1(v)\implies s_1\in\dgamma(i(v))
		\right\}.
	\]
	\item The top ($\itop$) and bottom ($\ibot$) elements:
	\begin{align*}
		\itop&\sdef\lambda v\in\VP:\dtop\\
		\ibot&\sdef\lambda v\in\VP:\dbot.
	\end{align*}
	\item The join operator $\ijoin:\I\times\I\to\I$, defined component-wise:
	\[
		i_1\ijoin i_2\sdef \lambda v\in\VP: i_1(v)\djoin i_2(v).
	\]
	
	\item The meet operator $\imeet:\I\times\I\to\I$, also defined component-wise:
	\[
		i_1\imeet i_2\sdef\lambda v\in\VP: i_1(v)\dmeet i_2(v).
	\]
	\item A function $\stabilise_N:\I\times\D\to\D$, satisfying:
	\[
		\dgamma(\stabilise_N(i,d))\supseteq\dgamma(d)\cup\{s_2\in\s\mid\exists s_1\in\dgamma(d):(s_1,s_2)\in\igamma(i)\}
	\]
	and defined as:
	\begin{align*}
		&\text{let }X=\bigdjoin_{V\in\powset(\V),\:|V|\leq N}
		\havoc{
			\left(
				d\dmeet\bigdmeet_{v\in V}
				i(v)
			\right)
		}{
			V
		}\text{ in}\\
		&\text{let }\V_N=\left\{V\in\powset(\V)\text{\LARGE{$\,\mid\,$}} |V|=N+1\land \left( d\dmeet\bigdmeet_{v\in V}i(v)\right) \ne\dbot\right\}\text{ in}\\
		&\text{let }Y=\havoc{
			\left(\ \ 
				\bigdjoin_{V\in\V_N}
				d\dmeet\bigdmeet_{v\in V}i(v)
			\right)
		}{
			\bigcup_{V\in\V_N}V
		}\text{ in}\\
		&d\djoin X\djoin Y.
	\end{align*}
	\item A function $\transitions:\D\times\assign\to\I$, satisfying:
	\[
		\igamma(\transitions(d,A))\supseteq\{(s_1,s_2)\in\s\times\s \mid
		s_1\in\dgamma(d)\land s_2=\sconc(A,s_1)
		\}
	\]
	and defined as:
	\[
		\transitions(d,\langle x_1,...,x_n:=e_1,...,e_n\rangle)\sdef
		\ibot[v\in\{x_1,...,x_n\}\mapsto d].
	\]
	\item A (weakening) transitive closure function $\close:\I\to\I$, satisfying:
	\begin{align*}
		&\close(i)\isqsupseteq i
		\\&\land\forall s_1,s_2,s_3\in\s:
		(s_1,s_2)\in\igamma(\close(i))
		\land(s_2,s_3)\in\igamma(\close(i))
		\\&\hspace{5.33em}\implies(s_1,s_3)\in\igamma(\close(i))
	\end{align*}
	and defined as:
	\[
		\close(i)\sdef\mu i:\lambda v: \bigdjoin_{V\in\powset(\VP)}\havoc{i(v)}{V}\dmeet\bigdmeet_{v'\in V}i(v').
	\]
\end{description}

\end{document}